\documentclass{elsarticle}

\usepackage{lineno,hyperref}
\modulolinenumbers[5]

\journal{Journal of Artificial Intelligence}

\usepackage{amssymb}
\usepackage{algorithm}
\usepackage[noend]{algpseudocode}
\usepackage{mathtools}
\usepackage{natbib}
\usepackage{pgf}
\usepackage[caption=false]{subfig}
\usepackage{appendix} 
\DeclarePairedDelimiter{\ceil}{\lceil}{\rceil}
\usepackage{tikz}

\newtheorem{theorem}{Theorem}
\newtheorem{lemma}[theorem]{Lemma}
\newdefinition{definition}{Definition}
\newproof{proof}{Proof}
\newproof{proof_lg}{Proof of Theorem \ref{thm:carving-equiv-tree}}

\begin{document}
\begin{frontmatter}
\title{Efficient Contraction of Large Tensor Networks for Weighted Model Counting through Graph Decompositions}

\author[1]{Jeffrey M. Dudek\corref{cor1}}
\ead{jmd11@rice.edu}
\author[2]{Leonardo Due{\~n}as-Osorio}
\ead{leonardo.duenas-osorio@rice.edu}
\author[1]{Moshe Y. Vardi}
\ead{vardi@rice.edu}

\address[1]{Department of Computer Science, Rice University, Houston, TX, USA}
\address[2]{Department of Civil and Environmental Engineering, Rice University, Houston, TX, USA}
\cortext[cor1]{Corresponding author}

\begin{abstract}
Constrained counting is a fundamental problem in artificial intelligence. A promising new algebraic approach to constrained counting makes use of tensor networks, following a reduction from constrained counting to the problem of tensor-network contraction. Contracting a tensor network efficiently requires determining an efficient order to contract the tensors inside the network, which is itself a difficult problem.

In this work, we apply graph decompositions to find contraction orders for tensor networks. 
We show that tree decompositions can be used both to find memory-efficient contraction orders and to factor tensor networks with high-rank, structured tensors. We implement these algorithms on top of state-of-the-art heuristic solvers for tree decompositions and show empirically that the resulting weighted model counter is quite effective and useful as part of a portfolio of counters.
\end{abstract}
\begin{keyword}
Weighted Model Counting \sep Tensor Network Contraction \sep Tree Decomposition \sep Carving Decomposition
\end{keyword}
\end{frontmatter}

\newcommand{\domain}[1]{[#1]}				
\newcommand{\dlabel}[1]{\ell(#1)} 			

\newcommand{\restrict}[1]{\big|_{#1}}		

\newcommand{\tdim}[1]{\mathcal{I}(#1)}		
\newcommand{\rank}[1]{rank(#1)}				
\newcommand{\size}[1]{size(#1)}				

\newcommand{\tntensor}[1]{\mathcal{T}(#1)}	
\newcommand{\tnfree}[1]{\mathcal{F}(#1)}	
\newcommand{\tnbound}[1]{\mathcal{B}(#1)}	

\newcommand{\edge}[1]{\{#1\}}				
\newcommand{\vinc}[2]{\delta_{#1}(#2)}		
\newcommand{\vincf}[1]{\delta_{#1}}		
\newcommand{\einc}[2]{\epsilon_{#1}(#2)}	
\newcommand{\eincf}[1]{\epsilon_{#1}}	    
\newcommand{\E}[1]{\mathcal{E}(#1)}			
\newcommand{\V}[1]{\mathcal{V}(#1)}			
\newcommand{\Lv}[1]{\mathcal{L}(#1)}		
\newcommand{\C}[3]{\mathcal{C}_{#1,#2}(#3)}	
\newcommand{\paritions}[1]{\mathcal{P}(#1)}
\newcommand{\Line}[1]{Line(#1)}

\newcommand{\blocks}[1]{\mathcal{B}(#1)}
\newcommand{\Ind}[0]{\mathbf{Ind}}
\newcommand{\pkg}[1]{\texttt{#1}}
\newcommand{\tool}[1]{\texttt{#1}}

\newcommand{\fv}[0]{z}
\newcommand{\copyt}[0]{\text{COPY}}

\newcommand{\support}[1]{\text{sup}(#1)}
\newcommand{\depend}[1]{\text{dep}(#1)}
\newcommand{\shortcite}[1]{\cite{#1}} 
\section{Introduction}
Constrained counting is a fundamental problem in artificial intelligence, with applications in probabilistic reasoning, planning, inexact computing, engineering reliability, and statistical physics \cite{Bacchus2003,DH07,GSS08}. In constrained counting (also called weighted model counting) the task is to count the total weight, subject to a given weight function, of the set of solutions of input constraints. Even when the weight function is a constant function, constrained counting is \#P-Complete \cite{Valiant79}. Nevertheless, the development of tools that can successfully compute the total weight on large industrial formulas is an area of active research \cite{OD15,Thurley2006}. 

Constrained counting can be reduced to the problem of tensor-network contraction \cite{BMT15}. \emph{Tensor networks} are a tool used across quantum physics and computer science for describing and reasoning about quantum systems, big-data processing, and more \cite{BB17,Cichocki14,Orus19}. A tensor network describes a complex tensor as a computation on many simpler tensors, and the problem of tensor-network \emph{contraction} is to perform this computation. Although tensor networks can be seen as a variant of factor graphs \cite{KFL01}, working directly with tensor networks allows us to leverage massive practical work in machine learning and high-performance computing on tensor contraction \cite{BK07,Hirata03,KKCLA17,VZTGDMVAC18} (which also includes GPU support \cite{KSTKPPRS19,NRBHHJN15}) to perform constrained counting. Since tensor networks are relatively unknown in the artificial intelligence community, we give an introduction of relevant material on tensors and tensor networks in Section \ref{sec:tensors}.

Contracting a tensor network requires determining an order to contract the tensors inside the network, and so efficient contraction requires finding a contraction order that minimizes computational cost. Since the number of possible contraction orders grows exponentially in the number of tensors, cost-based exhaustive algorithms, e.g. \cite{PHV14}, cannot scale to handle the large tensor networks required for the reduction from constrained counting. Instead, recent work \cite{KCMR18} gave heuristics that can sometimes find a ``good-enough'' contraction order through structure-based optimization. Finding efficient contraction orders for tensor networks remains an area of active research \cite{RTPCTSL19}.

The primary contribution of this work is the application of heuristic graph-decomposition techniques to find efficient contraction orders for tensor networks. Algorithms based on graph decompositions have been successful across computer science \cite{GLST17,MPPV04}, and their success in practice relies on finding good decompositions of arbitrary graphs. This, along with several recent competitions \cite{DKTW18}, has spurred the development of a variety of heuristics and tools for efficiently finding graph decompositions \cite{AMW17,HS18,Tamaki17}. While we do not establish new parameterized complexity results for model counting (as fixed-parameter algorithms for model counting are well-known for a variety of parameters \cite{FMR08,SS10}), we combine these theoretical results with high-performance tensor network libraries and with existing heuristic graph-decomposition tools to produce a competitive tool for weighted model counting.

We first discuss the \textbf{Line-Graph} method (\textbf{LG}) for finding efficient contraction orders through structure-based graph analysis. First applied to tensor networks by Markov and Shi \shortcite{MS08}, we contribute a new analysis that more closely matches the memory usage of existing tensor libraries. Our analysis combines two theoretical insights: (1) memory-efficient contraction orders are equivalent to low-width carving decompositions (first observed in \cite{de15}), and (2) tree decompositions can be used to find carving decompositions. \textbf{LG} has previously been implemented using exact tools for finding tree decompositions \cite{DFGHSW18}, but its implementation using heuristic tools for tree decompositions is largely unexplored.

Although \textbf{LG} is a general-purpose method for finding contraction orders, \textbf{LG} cannot handle high-rank tensors and so cannot solve many existing counting benchmarks. We therefore contribute a novel structure-based method for finding efficient contraction orders, tailored for constrained counting: the \textbf{Factor-Tree} method (\textbf{FT}). \textbf{FT} factors high-rank, highly-structured tensors as a preprocessing step, leveraging prior work on Hierarchical Tucker representations \cite{Grasedyck10}.

In order to compare our approaches against other model counters (\tool{cachet} \cite{SBK05}, \tool{miniC2D} \cite{OD15}, \tool{d4} \cite{LM17}, \tool{dynQBF} \cite{CW16}, \tool{dynasp} \cite{FHMW17} and \tool{SharpSAT} \cite{Thurley2006}) and other tensor-based methods, we implement \textbf{LG} and \textbf{FT} using three state-of-the-art heuristic tree-decomposition solvers in \tool{TensorOrder}, a new weighted model counter. \textbf{LG} outperforms other model counters and tensor-based methods on a set of unweighted benchmarks, while \textbf{FT} improves the virtual best solver on 21\% of a standard set of weighted model counting benchmarks. \tool{TensorOrder} is thus useful as part of a portfolio of weighted model counters. All code, benchmarks, and detailed data of benchmark runs are available at \url{https://github.com/vardigroup/TensorOrder}.

The rest of the paper is organized as follows: we provide graph notations and define graph decompositions in Section~\ref{sec:prelim}. We introduce tensors and tensor networks and discuss prior work on the optimization of tensor-network contraction in Section~\ref{sec:tensors}. We introduce a framework for solving the problem of weighted model counting with tensor networks in Section~\ref{sec:wmc}. We discuss the \textbf{Line-Graph} method in Section~\ref{sec:contraction-theory} and the \textbf{Factor-Tree} method in Section~\ref{sec:preprocessing}. We present an experimental evaluation of tensor-based approaches to model counting in Section~\ref{sec:experiments}. Finally, we discuss future work and conclude in Section~\ref{sec:conclusion}.
\section{Preliminaries: Graph Notations}
\label{sec:prelim}
A \emph{graph} $G$ has a nonempty set of vertices $\V{G}$, a set of (undirected) edges $\E{G}$, a function $\delta_G: \V{G} \rightarrow 2^{\E{G}}$ that gives the set of edges incident to each vertex, and a function $\epsilon_G: \E{G} \rightarrow 2^{\V{G}}$ that gives the set of vertices incident to each edge. Each edge must be incident to exactly two vertices, but multiple edges can exist between two vertices. An \emph{edge clique cover} of a graph $G$ is a set $A \subseteq 2^{\V{G}}$ such that (1) every vertex $v \in \V{G}$ is an element of some set in $A$, and (2) every element of $A$ is a clique in $G$ (that is, for every $C \in A$ and every pair of distinct $v, w \in C$ there is an edge between $v$ and $w$ in $G$).

A \emph{tree} is a simple, connected, and acyclic graph. A \emph{leaf} of a tree $T$ is a vertex of degree one, and we use $\Lv{T}$ to denote the set of leaves of $T$. A \emph{rooted binary tree} is a tree $T$ where either $T$ consists of a single vertex (called the \emph{root}), or every vertex of $T$ has degree one or three except a single vertex of degree two (called the \emph{root}). If $|\V{T}| > 1$, the \emph{immediate subtrees of $T$} are the two rooted binary trees that are the connected components of $T$ after the root is removed. Throughout this work, we often refer to a vertex of a tree as a \emph{node} and an edge as an \emph{arc} to avoid confusion, since our proofs will frequently work simultaneously with a graph and an associated tree.

In this work, we use two decompositions of a graph as a tree: carving decompositions \cite{ST94} and tree decompositions \cite{RS91}. Both decompose the graph into an \emph{unrooted binary tree}, which is a tree where every vertex has degree one or three. First, we describe carving decompositions \cite{ST94}:
\begin{definition}[Carving Decomposition]
\label{def:carving}
	Let $G$ be a graph. A \emph{carving decomposition} for $G$ is an unrooted binary tree $T$ whose leaves are the vertices of $G$, i.e. $\Lv{T} = \V{G}$. 
	
    For every arc $a$ of $T$, deleting $a$ from $T$ yields exactly two trees, whose leaves define a partition of the vertices of $G$. Let $C_a \subseteq \V{G}$ be an arbitrary element of this partition. The \emph{width} of $T$, denoted $width_c(T)$, is the maximum number of edges in $G$ between $C_a$ and $\V{G} \setminus C_a$ for all $a \in \E{T}$, i.e.,
    
	$$width_c(T) = \max_{a \in \E{T}} \left| \left( \bigcup_{v \in C_a} \vinc{G}{v} \right) \cap \left( \bigcup_{v \in \V{G} \setminus C_a} \vinc{G}{v} \right) \right|.$$
	

	
    The width of a carving decomposition $T$ with no edges is 0.
\end{definition}

The \emph{carving width} of a graph $G$ is the minimum width across all carving decompositions for $G$. Note that an equivalent definition of carving decompositions allows for degree two vertices within the tree as well.

Next, we define tree decompositions \cite{RS91}:
\begin{definition}[Tree Decomposition]
	Let $G$ be a graph. A \emph{tree decomposition} for $G$ is an unrooted binary tree $T$ together with a labeling function $\chi : \V{T} \rightarrow 2^{\V{G}}$ that satisfies the following three properties:
	\begin{enumerate}
		\item Every vertex of $G$ is contained in the label of some node of $T$. That is, $\V{G} = \bigcup_{n \in \V{T}} \chi(n)$.
		\item For every edge $e \in \E{G}$, there is a node $n \in \V{T}$ whose label is a superset of $\einc{G}{e}$, i.e. $\einc{G}{e} \subseteq \chi(n)$.
		\item If $n$ and $o$ are nodes in $T$ and $p$ is a node on the path from $n$ to $o$, then $\chi(n) \cap \chi(o) \subseteq \chi(p)$.
	\end{enumerate}
	The \emph{width} of a tree decomposition, denoted $width_t(T, \chi)$, is the maximum size (minus 1) of the label of every node, i.e.,
	$$width_t(T, \chi) = \max_{n \in \V{T}} | \chi(n) | - 1.$$
\end{definition}

The \emph{treewidth} of a graph $G$ is the minimum width across all tree decompositions for $G$. The treewidth of a tree is $1$. Treewidth is bounded by thrice the carving width \cite{sasak10}. Carving decompositions are the dual of branch decompositions, which are closely related to tree decompositions \cite{RS91}.
\section{An Introduction to Tensors and Tensor Networks}
\label{sec:tensors}
In this section, we introduce tensors and tensor networks and discuss prior work on the optimization of tensor-network contraction. To aid in exposition, along the way we build an analogy between the language of databases \cite{SG88}, the language of factor graphs \cite{KFL01,dechter99}, and the language of tensors: see Table \ref{table:db-tensor-analogy}.

\begin{table}[t]
\centering
\begin{tabular}{c|c|c}
\hline
\textbf{Database Concept} & \textbf{Factor Graph Concept} & \textbf{Tensor Concept}\\ \hline
Attribute & Variable & Index\\
Table & Factor & Tensor\\
Project-Join Query & Factor Graph & Tensor Network\\
Join Tree & Elimination Order & Contraction Tree\\ \hline
\end{tabular}
\caption{\label{table:db-tensor-analogy} An analogy between the language of databases, the language of factor graphs, and the language of tensors.}
\end{table}

\subsection{Tensors}
\emph{Tensors} are a generalization of vectors and matrices to higher dimensions-- a tensor with $r$ dimensions is a table of values each labeled by $r$ indices. An index is analogous to a variable in constraint satisfaction or an attribute in database theory. 

Fix a set $\Ind$ and define an \emph{index} to be an element of $\Ind$. For each index $i$ fix a finite set $\domain{i}$ called the \emph{domain} of $i$. An \emph{assignment} to a set of indices $I \subseteq \Ind$ is a function $\tau$ that maps each index $i \in I$ to an element of $\domain{i}$. Let $\domain{I}$ denote the set of assignments to $I$, i.e., $$\domain{I} = \{\tau: I \rightarrow \bigcup_{i \in I} \domain{i}~\text{s.t.}~\tau(i) \in \domain{i}~\text{for all}~i \in I\}.$$


We now define tensors as multidimensional arrays of values, indexed by assignments to a set of indices:\footnote{In some works, a tensor is defined as a multilinear map and Definition \ref{def:tensor} would be its representation in a fixed basis.}
\begin{definition}[Tensor] \label{def:tensor}
	A \emph{tensor} $A$ over a finite set of indices (denoted $\tdim{A}$) is a function $A: \domain{\tdim{A}} \rightarrow \mathbb{C}$ (where $\mathbb{C}$ is the set of complex numbers).
\end{definition}

The \emph{rank} of a tensor $A$ is the cardinality of $\tdim{A}$. The memory to store a tensor (in a dense way) is exponential in the rank. For example, a scalar is a rank 0 tensor, a vector is a rank 1 tensor, and a matrix is a rank 2 tensor. An example of a higher-rank tensor is the \emph{copy tensor} on a set of indices $I$, which is the tensor $\copyt_I: \domain{I} \rightarrow \mathbb{C}$ such that, for all $\tau \in \domain{I}$, $\copyt_I(\tau) \equiv 1$ if $\tau$ is a constant function on $I$ and $\copyt_I(\tau) \equiv 0$ otherwise \cite{BCJ11}.

It is common to consider sets of tensors closed under contraction (see Section \ref{sec:tensors:tensor-networks}), e.g. tensors with entries in $\mathbb{R}$ as in Section \ref{sec:wmc}. Database tables under bag-semantics \cite{CV93}, i.e., multirelations, are tensors with entries in $\mathbb{N}$. Probabilistic database tables \cite{CP87} are tensors with entries in $[0, 1]$ that sum to 1.

Many tools exist (e.g. \pkg{numpy} \cite{numpy}) to efficiently manipulate tensors. In Section \ref{sec:experiments}, we use these tools to implement tensor-network contraction, defined next.



\subsection{Tensor Networks}
\label{sec:tensors:tensor-networks}
A \emph{tensor network} defines a complex tensor by combining a set of simpler tensors in a principled way. This is analogous to how a database query defines a resulting table in terms of a computation across many tables.

\begin{definition}[Tensor Network]
	\label{def:tensor-contraction-network}
	A \emph{tensor network} $N$ is a nonempty set of tensors across which no index appears more than twice.
\end{definition}

\emph{Free indices} of $N$ are indices that appear once, while \emph{bond indices} of $N$ are indices that appear twice. We denote the set of free indices of $N$ by $\tnfree{N}$ and the set of bond indices of $N$ by $\tnbound{N}$. The \emph{bond dimension} of $N$ is the maximum size of $\domain{i}$ for all bond indices $i$ of $N$.

The problem of \emph{tensor-network contraction}, given an input tensor network $N$, is to compute the \emph{contraction} of $N$ by marginalizing all bond indices:
\begin{definition}[Tensor Network Contraction]
The \emph{contraction} of a tensor network $N$ is a tensor $\tntensor{N}$ with indices $\tnfree{N}$ (the set of free indices of $N$), i.e. a function $\tntensor{N} : \domain{\tnfree{N}} \rightarrow \mathbb{C}$, that is defined for all $\tau \in \domain{\tnfree{N}}$ by
		\begin{equation}
        \label{eqn:contraction} 
        \tntensor{N}(\tau) \equiv \sum_{\rho \in \domain{\tnbound{N}}} \prod_{A \in N} A((\rho \cup \tau)\restrict{\tdim{A}}).
        \end{equation}
\end{definition}

For example, the contraction of the tensor network $\{\copyt_I, \copyt_J\}$ is the tensor $\copyt_{I \oplus J}$ (where $I \oplus J$ is the symmetric difference of $I$ and $J$). Notice that if a tensor network has no free indices then its contraction is a rank 0 tensor. We write $A \cdot B$ to mean the contraction of the tensor network containing the two tensors $A$ and $B$. 

Following our analogy, given a tensor network containing database tables (under bag-semantics) as tensors, its contraction is the join of those tables followed by the projection of all shared attributes. Thus a tensor network is analogous to a project-join query. A tensor network can also be seen as a variant of a factor graph \cite{KFL01} with the additional practical restriction that no variable appears more than twice. The contraction of a tensor network corresponds to the marginalization of a factor graph \cite{RS17} and can similarly be seen as a special case of the FAQ problem \cite{KNR16}. The restriction on the appearance of variables is heavily exploited in tools for tensor-network contraction, since it allows tensor contraction to be implemented as matrix multiplication and thus leverage significant work in high-performance computing on matrix multiplication, both on the CPU \cite{LHKK77} and the GPU \cite{FSH04}.

We focus in this work on tensor networks with relatively few (or no) free indices and hundreds or thousands of bond indices. Such tensor networks are obtained in a variety of applications \cite{Cichocki14,DLVR18}, including the reduction from model counting to tensor network contraction \cite{BMT15}. Although the rank of the contraction $\tntensor{N}$ is small in this case, computing entries by
directly following Equation \ref{eqn:contraction} requires performing a summation over an exponential number of terms--- one for each assignment in $\domain{\tnbound{N}}$--- and is therefore infeasible.

$\tntensor{N}$ can instead be computed by recursively decomposing the tensor network, as in Algorithm \ref{alg:network-contraction} \cite{EP14}. The choice of rooted binary tree $T$ does not affect the output of Algorithm \ref{alg:network-contraction} but may have a dramatic impact on the running-time and memory usage. We explore this in more detail in the following section.

\begin{algorithm}[t]
	\caption{Recursively contracting a tensor network}\label{alg:network-contraction}
	\hspace*{\algorithmicindent} \textbf{Input:} A tensor network $N$ and a rooted binary tree $T$ whose leaves are the tensors of $N$, i.e. $\Lv{T} = N$. \\
	\hspace*{\algorithmicindent} \textbf{Output:} $\tntensor{N}$, the contraction of $N$.
	\begin{algorithmic}[1]
	    \Procedure{Contract}{$N,T$}
		\If {$\left|N\right| = 1$}
		\State \Return the tensor contained in $N$
		\Else
        \State $T_1, T_2 \gets \text{immediate subtrees of}~T$
		\State $A_1 \gets \Call{Contract}{\Lv{T_1}, T_1}$
		\State $A_2 \gets \Call{Contract}{\Lv{T_2}, T_2}$
		\State \Return $A_1 \cdot A_2$
		\EndIf
		\EndProcedure
	\end{algorithmic}
\end{algorithm}

\subsection{Contracting Tensor Networks}
The rooted binary trees used by Algorithm \ref{alg:network-contraction} are called contraction trees:
\begin{definition}[Contraction Tree \cite{EP14}] \label{def:contraction-tree}
	Let $N$ be a tensor network. A \emph{contraction tree} for $N$ is a rooted binary tree $T$ whose leaves are the tensors of $N$. 
\end{definition}

In our analogy, a contraction tree for a tensor network representing a project-join query is a join tree of that query (with projections done as early as possible). 

The problem of \emph{tensor-network-contraction optimization}, which we tackle in this paper, is given a tensor network $N$ to find a contraction tree that minimizes the computational cost of Algorithm \ref{alg:network-contraction}. 
Several \emph{cost-based} approaches aim for minimizing the total number of floating point operations to perform Algorithm \ref{alg:network-contraction} in step 8, e.g. \cite{PHV14} and the \pkg{einsum} package in \pkg{numpy}.  
In this work, we instead focus on \emph{structure-based} approaches to tensor-network-contraction optimization, which analyze the rank of intermediate tensors that appear during Algorithm \ref{alg:network-contraction}. These ranks indicate the amount of memory and computation required at each recursive stage. Moreover, these ranks are more amenable to analysis.

One line of work \cite{MS08,DFGHSW18} uses graph decompositions to analyze the \emph{contraction complexity} of a contraction tree: the maximum over all recursive calls of the sum (minus 1) of the rank of the two tensors contracted in step 8 of Algorithm \ref{alg:network-contraction}. Contraction complexity measures the memory required when step 8 is computed by summing over each shared index sequentially. However, modern tensor packages (e.g. \pkg{numpy}) instead sum over all shared indices simultaneously, which requires the same number of floating point operations but often requires significantly less intermediate memory. Thus contraction complexity overestimates the memory requirements of many contraction trees. 

Instead, another line of structure-based optimization analyzes the maximum rank over all recursive calls of the result of step 8 (and step 3). We call this the \emph{max rank} of a contraction tree. Max rank measures the memory required when step 8 is computed by summing over all shared indices simultaneously. Thus max rank estimates the memory usage of modern tensor packages. Recent work \cite{KCMR18} introduced three methods for heuristically minimizing the max rank: 
a greedy approach (called \textbf{greedy}), an approach using graph partitioning (called \textbf{metis}), and an approach using community-structure detection (called \textbf{GN}).

In this work, we improve on these methods by using graph decompositions to find contraction trees with small max rank.

\section{From Weighted Model Counting to Tensor Networks}
\label{sec:wmc}
In this section, we introduce a framework for solving the problem of weighted model counting with tensor networks. The task in weighted model counting is to count the total weight, subject to a given (literal) weight function, of the set of solutions of input constraints. Formally:
\begin{definition}[Weighted Model Count]
  Let $\varphi$ be a formula over Boolean variables $X$ and let $W: X \times \{0,1\} \rightarrow \mathbb{R}$ be a function (called the \emph{weight function}). The \emph{weighted model count} of $\varphi$ w.r.t. $W$ is
  $$W(\varphi) \equiv \sum_{\tau \in \domain{X}} \varphi(\tau) \cdot \prod_{x \in X} W(x, \tau(x)).$$
\end{definition}




Note that $[X]$ is the set of all functions $\tau$ from $X$ to $\{0, 1\}$. Existing reductions from model counting to tensor-network contraction \cite{BMT15,KCMR18} focus on the unweighted case (i.e., when $W$ is constant). Since we are interested in weighted model counting, we prove that the reduction can be easily extended:
\begin{theorem}
\label{thm:wmc-reduction}
Let $\varphi$ be a CNF formula over Boolean variables $X$ and let $W$ be a weight function. One can construct in polynomial time a tensor network $N_\varphi$ such that $\tnfree{N_\varphi} = \emptyset$ and the contraction of $N_\varphi$ is $W(\varphi)$.
\end{theorem}
\begin{proof}
The key idea is to include in $N_\varphi$ a tensor $A_x$ for each variable $x \in X$ and a tensor $B_C$ for each clause $C \in \varphi$ such that the tensors share an index if and only if the corresponding variable appears in the corresponding clause.

For each $C \in \varphi$, let $\support{C}$ be the set of Boolean variables that appear in $C$. Similarly, for each $x \in X$, let $\depend{x}$ be the set of clauses that contain $x$. Define $I = \{(x, C) : C \in \varphi, x \in \support{C}\}$ to be a set of indices, each with domain $\{0, 1\}$. That is, $I$ has an index for each appearance of each variable in $\varphi$.

For each $x \in X$, let $A_x: [\{x\} \times \depend{x}] \rightarrow \mathbb{R}$ be the tensor defined by $$A_x(\tau) \equiv \begin{cases}W(x, 1)&\text{if}~\tau((x,C)) = 1~\text{for all}~C \in \depend{x}\\W(x, 0)&\text{if}~\tau((x,C)) = 0~\text{for all}~C \in \depend{x}\\0&\text{otherwise.}\end{cases}$$

Next, for each $C \in \varphi$, let $B_C: [\support{C} \times \{C\}] \rightarrow \mathbb{R}$ be the tensor defined by
$$B_C(\tau) \equiv \begin{cases}1& \text{if}~\{x: x \in \support{C}~\text{and}~\tau((x, C)) = 1\}~\text{satisfies}~C\\0&\text{otherwise.}\end{cases}$$

Finally, let $N_\varphi = \{ A_x : x \in X\} \cup \{B_C : C \in \varphi\}$. Each index $(x, C) \in I$ appears exactly twice in $N_\varphi$ (in $A_x$ and $B_C$) and so is a bond index. Thus $N_\varphi$ is a tensor network with $\tnfree{N_\varphi} = \emptyset$. We now compute $\tntensor{N_\varphi}$. By Definition 5,
$$\tntensor{N_\varphi}(\emptyset) = \sum_{\rho \in \domain{I}} \prod_{x \in X} A_x(\rho\restrict{\{x\} \times \depend{x}}) \cdot \prod_{C \in \varphi} B_C(\rho\restrict{\support{C} \times \{C\}}).$$

To compute the term of this sum for each $\rho \in \domain{I}$, we examine if there exists some $\tau_\rho \in \domain{X}$ such that $\rho((x, C)) = \tau_\rho(x)$ for all $(x, C) \in I$. If so, then by construction $A_x(\rho\restrict{\{x\} \times \depend{x}}) = W(x, \tau_\rho(x))$ and $B_C(\rho\restrict{\support{C} \times \{C\}}) = C(\tau_\rho\restrict{\support{C}})$. On the other hand, if no such $\tau_\rho$ exists then there is some $y \in X$ such that $\rho$ is not constant on $\{y\} \times \depend{y}$. Thus by construction $A_y(\rho\restrict{\{y\} \times \depend{y}}) = 0$ and so the term in the sum for $\rho$ is 0. Hence
$$\tntensor{N_\varphi}(\emptyset) = \sum_{\tau \in \domain{X}} \prod_{x \in X} W(x, \tau(x)) \cdot \prod_{C \in \varphi} C(\tau\restrict{\support{C}}) = W(\varphi).$$\hfill$\square$
\end{proof}

\begin{figure}[t]
	\centering
	\begin{tikzpicture}
\begin{scope}[every node/.style={circle,thick,draw}]
    \node (C) at (-1,-1) {$C$};
    \node (x) at (-1,0) {$x$};
    \node (A) at (-1,1) {$A$};
    \node (y) at (0,0) {$y$};
    \node (w) at (0,1) {$w$};
    \node (D) at (1,-1) {$D$};
    \node (z) at (1,0) {$z$};
    \node (B) at (1,1) {$B$};
\end{scope}

\begin{scope}[every node/.style={fill=white,circle},
              every edge/.style={draw=black,very thick}]
    \path [-] (A) edge (w);
    \path [-] (A) edge (x);
    \path [-] (A) edge (y);
    \path [-] (B) edge (w);
    \path [-] (B) edge (y);
    \path [-] (B) edge (z);
    \path [-] (C) edge (x);
    \path [-] (C) edge (y);
    \path [-] (D) edge (y);
    \path [-] (D) edge (z);
\end{scope}
\end{tikzpicture}
	\hspace{1cm}
	\begin{tikzpicture}
    \node (y) at (-1.5,1) {$y$};
    \node (C) at (-1.5,0.33) {$C$};
    \node (x) at (-1.5,-0.33) {$x$};
    \node (A) at (-1.5,-1) {$A$};
    
    \node (w) at (1.5,1) {$w$};
    \node (B) at (1.5,0.33) {$B$};
    \node (D) at (1.5,-0.33) {$D$};
    \node (z) at (1.5,-1) {$z$};
    \node (yC) at (-0.75, 0.67) {};
    \node (xA) at (-0.75, -0.67) {};
    \node (wB) at (0.75, 0.67) {};
    \node (zD) at (0.75, -0.67) {};
    
    \node (yCxA) at (-0.25, 0) {};
    \node (wBzD) at (0.25, 0) {};

\begin{scope}[every node/.style={fill=black,rectangle}]
    \node (root) at (0, 0) {};
\end{scope}
    
\begin{scope}[every node/.style={fill=white,circle},
              every edge/.style={draw=black,very thick}]
    \path [-] (y) edge (yC.center);
    \path [-] (C) edge (yC.center);
    \path [-] (x) edge (xA.center);
    \path [-] (A) edge (xA.center);
    \path [-] (w) edge (wB.center);
    \path [-] (B) edge (wB.center);
    \path [-] (z) edge (zD.center);
    \path [-] (D) edge (zD.center);
    
    \path [-] (yC.center) edge (yCxA.center);
    \path [-] (xA.center) edge (yCxA.center);
    \path [-] (wB.center) edge (wBzD.center);
    \path [-] (zD.center) edge (wBzD.center);
    
    \path [-] (yCxA.center) edge (root);
    \path [-] (wBzD.center) edge (root);
\end{scope}
\end{tikzpicture}
	\caption{\label{fig:wmc-example} The tensor network (left) produced by Theorem \ref{thm:wmc-reduction} on $\varphi = (w \lor x \lor \neg y) \land (w \lor y \lor z) \land (\neg x \lor \neg y) \land (\neg y \lor \neg z)$, consisting of 8 tensors and 10 indices. Vertices in this diagram are tensors, while edges indicate that the tensors share an index. The weight function affects the entries of the tensors for $w$, $x$, $y$, and $z$. This tensor network has a contraction tree (right) of max rank 4, but no contraction trees of smaller max rank.}
\end{figure}
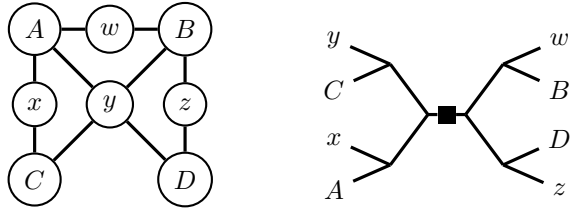

See Figure \ref{fig:wmc-example} for an example of the reduction. This reduction is closely related to the formulation of model counting as the marginalization of a factor graph representing the constraints. Unlike the reduction to factor-graph marginalization, which only assigns factors to clauses, we must also assign a tensor to each variable $x$. For example, if $x$ has weights $W(x, 0) = W(x,1) = 1$ then the tensor assigned to $x$ is a copy tensor. This reduction can also be extended beyond OR clauses to other types of constraints (e.g. parity or cardinality constraints). 

Theorem \ref{thm:wmc-reduction} suggests that the weighted model count of $\varphi$ can be computed by constructing and contracting $N_\varphi$. We present this framework as Algorithm \ref{alg:wmc}. Algorithm \ref{alg:wmc} is a fixed-parameter algorithm for model counting, parameterized by carving-width of the incidence graph. The existence of such algorithms is easily implied by fixed-parameter algorithms for model counting parameterized by treewidth \cite{FMR08,SS10} since treewidth is bounded by thrice the carving width \cite{sasak10}. A variety of methods can be used in Step 2 to find a contraction tree to contract $N_\varphi$, including the methods \textbf{LG} and \textbf{FT} that we discuss in the following sections.

\begin{algorithm}[t]
	\caption{Computing the weighted model count with a TN}\label{alg:wmc}
	\hspace*{\algorithmicindent} \textbf{Input:} A CNF formula $\varphi$ and a weight function $W$\\
	\hspace*{\algorithmicindent} \textbf{Output:} $W(\varphi)$, the weighted model count of $\varphi$ w.r.t. $W$
	\begin{algorithmic}[1]
	    \State $N_\varphi \gets \text{tensor network constructed via Theorem \ref{thm:wmc-reduction}}$
	    \State $T \gets \Call{Find\_Contraction\_Tree}{N_\varphi}$ \Comment{e.g., \textbf{LG} or \textbf{FG}}
	    \State \Return $\Call{Contract}{N_\varphi,~T}$
	\end{algorithmic}
\end{algorithm}


\section{The Line-Graph Method for Finding Contraction Trees}
\label{sec:contraction-theory}
The \textbf{Line-Graph} method for finding contraction trees for a tensor network $N$ applies graph-decomposition techniques to a particular graph constructed from $N$. Prior work \cite{MS08} on tensor networks with no free indices constructed a graph from a tensor network where tensors correspond to vertices and indices shared between tensors correspond to edges. In the context of constraint networks \cite{dechter03}, this is analogous to the dual constraint graph (if multiple edges are drawn between constraints with multiple variables in common).

Although tensor networks constructed from weighted model counting instances do not have free indices, we utilize tensor networks with free indices as part of the preprocessing in Section \ref{sec:preprocessing} and so we need a more general graph construction that can handle free indices. Other works, e.g. \cite{Ying17}, extend the graph construction of \cite{MS08} to tensor networks with free indices by treating free indices as ``half-edges'' (i.e., edges incident to one vertex), but decompositions of such graphs are not well-studied. In order to cleanly extend our decomposition-based analysis to tensor networks with free indices, in this work we instead add an extra vertex incident to all free indices, which we call the \emph{free vertex}. We call the resulting graph the \emph{structure graph} of a tensor network:
\begin{definition}[Structure Graph]\label{def:structure}
	Let $N$ be a tensor network. The \emph{structure graph} of $N$ is the graph $G$ whose 
    vertices are the tensors of $N$ and a fresh vertex $\fv$ (called the \emph{free vertex}) and whose edges are the indices of $N$. Each tensor is incident to its indices, and $\fv$ is incident to all free indices.
    That is, $\V{G} = N \sqcup \{ \fv \}$, $\E{G} = \tnbound{N} \cup \tnfree{N}$, $\vinc{G}{A} = \tdim{A}$ for all $A \in N$, and $\vinc{G}{\fv} = \tnfree{N}$.
\end{definition}

If $N$ has no free indices, the free vertex has no incident edges and the remaining graph is exactly the graph analyzed in prior work. Intuitively, the structure graph of a tensor network $N$ captures how indices are shared by the tensors of $N$. For example, on a CNF formula $\varphi$ Theorem \ref{thm:wmc-reduction} produces a tensor network $N_\varphi$ whose structure graph is exactly the \emph{incidence graph} of $\varphi$. The structure graph of $N$ contains all information needed to compute the max-rank of a contraction-tree of $N$, as formalized in the following lemma.
\begin{lemma} \label{lemma:tcn-equiv-structure}
	Let $N$ be a tensor network with structure graph $G$. If $N' \subseteq N$ is nonempty, then $N'$ is a tensor network and $\tnfree{N'} = \vinc{G}{N'} \cap \vinc{G}{\V{G} \setminus N'}$.
\end{lemma}
\begin{proof}
	$N'$ is a tensor network since $N$ is. Let $\fv$ be the free vertex of $G$.
	An index $i$ is free in $N'$ if and only if $i$ appears in $N'$ and either $i$ is free in $N$ or $i$ also appears in $N \setminus N'$. Thus
    $$\tnfree{N'} = \bigcup_{A \in N'} \tdim{A} \cap \left( \tnfree{N} \cup \bigcup_{B \in N \setminus N'} \tdim{B} \right).$$
	Since $\tdim{A} = \vinc{G}{A}$ for all $A \in N$ and $\tnfree{N} = \vinc{G}{z}$, we conclude that
	$$\tnfree{N'} = \vinc{G}{N'} \cap (\vinc{G}{\fv} \cup \vinc{G}{N \setminus N'}) =  \vinc{G}{N'} \cap \vinc{G}{\V{G} \setminus N'}.$$\hfill$\square$
\end{proof}


\subsection{Finding Contraction Trees from Carving Decompositions}
Contraction trees are closely connected to decompositions of the structure graph. In particular, contraction trees of a tensor network correspond to carving decompositions of its structure graph, where max rank corresponds exactly to carving width. This correspondence was first proven for tensor networks with no free indices by de Oliveira Oliveira \cite{de15}. Theorem \ref{thm:contraction-equiv-carving} extends this correspondence to tensor networks with free indices as well:
\begin{theorem}
	\label{thm:contraction-equiv-carving}
	Let $N$ be a tensor network with structure graph $G$ and let $w \in \mathbb{N}$. Then $N$ has a contraction tree of max rank $w$ if and only if $G$ has a carving decomposition of width $w$. Moreover, given one of these objects the other can be constructed in $O(|N|)$ time.
\end{theorem}
\begin{proof}
Let $\fv$ be the free vertex of $G$. First, let $T$ be a contraction tree of $N$ of max-rank $w$. Construct $T'$ from $T$ by adding $\fv$ as a leaf to the root of $T$. 

$T'$ is a carving decomposition of $G$, since $T'$ is an unrooted binary tree with $\Lv{T'} = \Lv{T} \sqcup \{\fv\} = \V{G}$. 
For each $a \in \E{T'}$, removing $a$ from $T'$ produces two connected components, both trees. Let $T_a$ be the connected component that does not contain $\fv$ and note that $T_a$ is a contraction tree for $\Lv{T_a} \subseteq N$. 

This gives us a bijection between $\E{T'}$ and the set of recursive calls of Algorithm \ref{alg:network-contraction}, where each $a \in \E{T'}$ corresponds to the recursive call where $T_a$ is the input contraction tree and $\tntensor{\Lv{T_a}}$ is the output tensor. Thus
\begin{align*}
&width_c(T') = \max_{a \in \E{T'}} \left| \vinc{G}{\Lv{T_a}} \cap \vinc{G}{\V{G} \setminus \Lv{T_a}} \right| = \max_{a \in \E{T'}} \left| \tnfree{\Lv{T_a}} \right| = w
\end{align*}
where the middle equality is given by applying Lemma \ref{lemma:tcn-equiv-structure} with $N' = \Lv{T_a}$.

Conversely, let $S$ be a carving decomposition of $G$ of width $w$. Construct $S'$ from $S$ by removing the leaf $\fv$ (and its incident arc) from $S$. $S'$ is a contraction tree of $N$, since $S$ is a rooted binary tree (whose root is the node previously attached by an arc to $\fv$) and $\Lv{S'} = \Lv{S} \setminus \{\fv\} = N$. Moreover, applying the construction in the first half of the proof produces $S$ and so the max rank of $S'$ is $width_c(S) = w$. \hfill$\square$
\end{proof}

One corollary of Theorem \ref{thm:contraction-equiv-carving} is that tensor networks with isomorphic structure graphs have contraction trees of equal max rank. This corollary is closely related to Theorem 1 of \cite{EP14}.

Carving decompositions have been studied in several settings. For example, there is an algorithm to find a carving decomposition of minimal width of a planar graph in time cubic in the number of edges \cite{GT08}. It follows that if the structure graph of a tensor network $N$ is planar, one can construct a contraction tree of $N$ of minimal max rank in time $O(|\tnbound{N} \cup \tnfree{N}|^3)$.

There is limited work on the heuristic construction of ``good'' carving decompositions for non-planar graphs. Instead, we leverage the work behind finding tree decompositions to find carving decompositions and subsequently find contraction trees of small max rank.

\subsection{Finding Contraction Trees from Tree Decompositions}


One technique for join-query optimization \cite{DKV02,MPPV04} focuses on analysis of the \emph{join graph}. The \emph{join graph} of a project-join query consists of all attributes of a database as vertices and all tables in the join as cliques. In this approach, tree decompositions for the join graph of a query are used to find optimal join trees. The analogous technique on factor graphs analyzes the \emph{primal graph} of a factor graph, which consists of all variables as vertices and all factors as cliques. Similarly, tree decompositions of the primal graph can be used to find variable elimination orders \cite{KDLD05}. The graph analogous to join graphs and primal graphs for tensor networks is the \emph{line graph} of the structure graph:
\begin{definition}[Line Graph]
	\label{def:line-graph}
	The \emph{line graph} of a graph $G$ is a graph $\Line{G}$ whose vertices are the edges of $G$, and where the number of edges between each $e,f \in \E{G}$ is $|\einc{G}{e} \cap \einc{G}{f}|$, the number of endpoints shared between $e$ and $f$.
\end{definition}

This technique was applied in the context of tensor networks by Markov and Shi \cite{MS08}, who proved that tree decompositions for $\Line{G}$ (where $G$ is the structure graph of a tensor network $N$) can be transformed into contraction trees for $N$ of small contraction complexity. Specifically, tree decompositions of optimal width $w$ yield contraction trees of contraction complexity $w+1$.

In the following theorem we analyze the max rank of the resulting contraction trees, which has not previously been studied. We present this result as a new relationship between carving width and treewidth:
\begin{theorem} \label{thm:carving-equiv-tree}
	Let $G$ be a graph with $\E{G} \neq \emptyset$. Given a tree decomposition $T$ for $\Line{G}$ of width $w \in \mathbb{N}$, one can construct in polynomial time a carving decomposition for $G$ of width no more than $w+1$.
\end{theorem}

In order to prove this theorem, it is helpful to first state and prove a lemma on simplifying the internal structure of a tree decomposition. In particular, we show that given an edge clique cover of a graph $G$, it is sufficient to only consider tree decompositions whose leaves are labeled only by elements of the edge clique cover.

\begin{lemma}\label{lemma:tree-simplification}
Let $G$ be a graph, let $(T, \chi)$ be a tree decomposition of $G$, and let $A$ be a finite set. If $f: A \rightarrow 2^{\V{G}}$ is a function whose image is an edge clique cover of $G$, then we can construct in polynomial time a tree decomposition $(S, \psi)$ of $G$ and a bijection $g: A \rightarrow \Lv{S}$ with $width_t(S, \psi) \leq width_t(T, \chi)$ and $\psi \circ g = f$.
\end{lemma}
\begin{proof}
Consider an arbitrary $a \in A$ and define $\chi_a: \V{T} \rightarrow 2^{\V{G}}$ by $\chi_a(n) \equiv \chi(n) \cap f(a)$ for all $n \in \V{T}$. Notice that $(T, \chi_a)$ is a tree decomposition of $G \cap f(a)$. Since $f$ is an edge clique cover, $G \cap f(a)$ is a complete graph with $|f(a)|$ vertices and thus has treewidth $|f(a)|-1$. It follows that $width_t(T, \chi_a) \geq |f(a)|-1$. That is, there is some $n_a \in \V{T}$ such that $|\chi_a(n_a)| \geq |f(a)|$. It follows that $\chi_a(n_a) = f(a)$ and so $f(a) \subseteq \chi(n_a)$. Choose an arbitrary arc $b \in \einc{T}{n_a}$ and construct $T'$ from $T$ by attaching a new leaf $\ell_a$ at $b$ (and introducing a new internal node). We can extend $\chi$ into a labeling $\chi': \V{T'} \rightarrow 2^{\E{G}}$ by labeling the new internal node with $\chi(n_a)$ and labeling the new leaf node with $f(a)$. Note that $(T', \chi')$ is still a tree decomposition of $G$ of width $width_t(T, \chi)$, that all labels of leaves of $(T, \chi)$ are still labels of leaves of $(T', \chi')$, and that the new leaf of $(T' \chi')$, namely $\ell_a$, is labeled by $f(a)$.

By repeating this process for every $a \in A$, by induction we produce in polynomial time a tree decomposition $(T', \chi')$ of width $width_t(T, \chi)$. Moreover, define the function $g: A \rightarrow \Lv{T'}$ for all $a \in A$ by $g(a) = \ell_a$ (the new leaf attached at each step). By construction, $\chi' \circ g = f$ (since each $\ell_a$ was labeled by $f(a)$) and moreover $f$ is an injection (since a new leaf was introduced at each step). It remains to make $f$ a bijection by removing leaves of $T'$.

Since $f$ is an edge clique cover, properties (1) and (2) of a tree decomposition can be satisfied purely looking at the nodes of $T'$ in the range of $g$. Moreover, removing leaves of $T'$ cannot falsify property (3) of a tree decomposition. Thus we can repeatedly remove leaves of $T'$ not in the range of $g$ until we eventually reach a tree decomposition $(S, \psi)$ for $G$ whose leaves are exactly the range of $g$. After this process, $g$ is a bijection as a function onto $\Lv{S}$ and $\psi \circ f = \delta_G$. Moreover, since $\psi(\V{S}) \subseteq \chi'(\V{T'})$ it follows that $width_t(S, \psi) \leq width_t(T', \chi') = width_t(T, \chi)$ as desired. \hfill$\square$
\end{proof}

We now use this lemma to prove Theorem \ref{thm:carving-equiv-tree}.

\begin{proof_lg}
First, observe that the image of $\vincf{G}: \V{G} \rightarrow \E{G}$ is an edge clique cover of $\Line{G}$. It follows by Lemma \ref{lemma:tree-simplification} that we can construct a tree decomposition $(S, \psi)$ of $\Line{G}$ and a bijection $g: \V{G} \rightarrow \Lv{S}$ such that $\psi \circ g = \vincf{G}$ and $width_t(S, \psi) \leq width_t(T, \chi)$.

Construct $T'$ from $S$ by replacing every leaf $\ell \in \Lv{S}$ with $g^{-1}(\ell)$. Since $g$ is a bijection, $\Lv{T'} = \V{G}$ and so $T'$ is a carving decomposition. In order to compute the carving width of $T'$, consider an arbitrary arc $a \in \E{T'}$ and let $C_a$ be an element of the partition of $\V{G}$ defined by removing $a$. 

For every edge $e \in \vinc{G}{C_a} \cap \vinc{G}{\Lv{T'} \setminus C_a}$, there must be vertices $v \in C_a$ and $w \in \Lv{T'} \setminus C_a$ that are both incident to $e$ and so $e \in \vinc{G}{v} \cap \vinc{G}{w}$. Since $\chi \circ g = \vincf{G}$, it follows that $e \in \chi(g(v)) \cap \chi(g(w))$.  Property 3 of tree decompositions implies that $e$ must also be in the label of every node in the path from $g(v)$ to $g(w)$ in $T$; in particular, $e$ must be in the label of both endpoints of $a$.

Thus every element of $\vinc{G}{C_a} \cap \vinc{G}{\Lv{T'} \setminus C_a}$ must be in the label of both endpoints of $a$. It follows that $|\vinc{G}{C_a} \cap \vinc{G}{\Lv{T'} \setminus C_a}| \leq width_t(T, \chi)+1$. Hence $width_c(T') \leq width_t(T, \chi)+1$ as desired. \hfill$\square$
\end{proof_lg}

An alternative proof of Theorem \ref{thm:carving-equiv-tree} uses Theorem 2.4 of \cite{HW18} to construct a carving decomposition $T'$ of $G$ from $T$ whose vertex congestion is $w+1$. By Lemma 2 of \cite{ACDJPS07}, it follows that the carving width of $T'$ is no more than $w+1$.

Applying Theorem \ref{thm:carving-equiv-tree} when $G$ is the structure graph of a tensor network (together with Theorem \ref{thm:contraction-equiv-carving}) gives us the \textbf{Line-Graph} method, which finds contraction trees by finding tree decompositions of the corresponding line graph. There are several advantages to our new analysis over the analysis of \cite{MS08}: our analysis holds for tensor networks with free indices, and we analyze the max rank of contraction trees instead of the contraction complexity. Although the contraction complexity (and, for factor graphs, the width of the elimination order) is equal to one plus the width of the used tree decomposition, the max rank is smaller on some graphs; see Section \ref{sec:experiments:graph_analysis} in the appendix for an experimental analysis of this.
\section{The Factor-Tree Method for Finding Contraction Trees}
\label{sec:preprocessing}
Approaches to tensor-network contraction that do not modify the input tensor network (e.g., \textbf{LG}) are inherently limited by the ranks of the input tensors. If a tensor network has a rank $r$ tensor, then all contraction trees have max rank of at least $r$. This is a problem for tensor networks with high-rank tensors. 

One example of tensor networks that may contain high-rank tensors are the networks obtained by the reduction from model counting. The tensor network produced from a formula $\varphi$ contains a tensor representing each variable $x$, where the rank of this tensor is the number of appearances of $x$ in $\varphi$ (e.g., the rank $4$ tensor for $y$ in Figure \ref{fig:wmc-example}). For many benchmarks, where a single variable might appear tens or even hundreds of times, this reduction will therefore produce tensor networks containing tensors of infeasibly-high rank. Reductions exist from model counting on arbitrary formulas to model counting on formulas where the number of appearances of each variables is small. However, existing reductions do not consider the carving width of the resulting incidence graph and so often do not significantly improve the max-rank of available contraction trees. 

We introduce here a novel method \textbf{Factor-Tree} that avoids this barrier by preprocessing the input tensor network. Our insight is that a tree decomposition for the incidence graph of $\varphi$ can be used as a guide to introduce new variables in a principled way, so that the resulting tensor network has good contraction trees. In the language of tensors, introducing new variables corresponds to \emph{factoring}: replacing each high-rank tensor $A$ with a tensor network $N_A$ of low-rank tensors that contracts to $A$. The key idea of \textbf{FT}, then, is to use a tree decomposition for the structure graph to factor high-rank tensors.

We state this new result as Theorem \ref{thm:factorable-tree}. Since not all tensors can be factored in the ways that we require for this theorem and for \textbf{FT}, we first characterize the required property: that every tensor is factorable as an arbitrary tree of tensors:



\begin{definition} \label{def:tree-factorable}
A tensor $A$ is \emph{tree factorable} if, for every tree $T$ whose leaves are $\tdim{A}$ (called a \emph{dimension tree} of $A$), there is a tensor network $N_A$ and a bijection $g_A: \V{T} \rightarrow N_A$ s.t.
\begin{enumerate}
\item $A$ is the contraction of $N_A$,
\item $g_A$ is an isomorphism between $T$ and the structure graph of $N_A$ with the free vertex (and incident edges) removed,
\item for every index $i$ of $A$, $i$ is an index of $g_A(i)$, and
\item for some index $i$ of $A$, the bond dimension of $N_A$ is no bigger than $|\domain{i}|$. 
\end{enumerate}
\end{definition}
All tensors in the reduction of Theorem \ref{thm:wmc-reduction} from weighted model counting to tensor networks are tree factorable. A tensor network $N_A$ that satisfies properties 1, 2, and 3 of Definition \ref{def:tree-factorable} for some tree is called a \emph{Hierarchical Tucker representation} of $A$ \cite{Grasedyck10}. Property 4 ensures the result of Theorem \ref{thm:factorable-tree} has small bond dimension.


We now state the main result of this section, which allows us to use a tree decomposition for the structure graph of a tensor network (containing only tree factorable tensors) to factor each tensor in the network and find a contraction tree of low max rank for the resulting network:
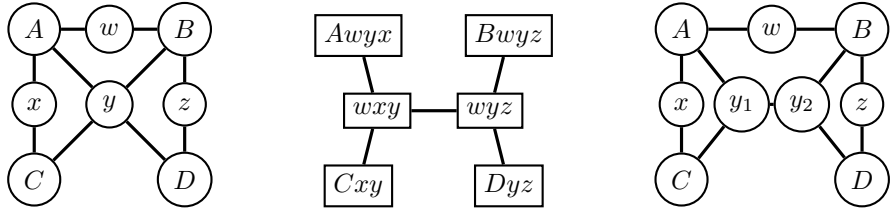
\begin{figure}[t]
	\centering
	\begin{tikzpicture}
\begin{scope}[every node/.style={circle,thick,draw}]
    \node (C) at (-1,-1) {$C$};
    \node (x) at (-1,0) {$x$};
    \node (A) at (-1,1) {$A$};
    \node (y) at (0,0) {$y$};
    \node (w) at (0,1) {$w$};
    \node (D) at (1,-1) {$D$};
    \node (z) at (1,0) {$z$};
    \node (B) at (1,1) {$B$};
\end{scope}

\begin{scope}[every node/.style={fill=white,circle},
              every edge/.style={draw=black,very thick}]
    \path [-] (A) edge (w);
    \path [-] (A) edge (x);
    \path [-] (A) edge (y);
    \path [-] (B) edge (w);
    \path [-] (B) edge (y);
    \path [-] (B) edge (z);
    \path [-] (C) edge (x);
    \path [-] (C) edge (y);
    \path [-] (D) edge (y);
    \path [-] (D) edge (z);
\end{scope}
\end{tikzpicture}
	\hspace{1cm}
	\begin{tikzpicture}
\begin{scope}[every node/.style={rectangle,thick,draw}]
    \node (1) at (-1,1) {$Awyx$};
    \node (2) at (-1,-1) {$Cxy$};
    \node (3) at (-0.75,0) {$wxy$};
    \node (4) at (0.75,0) {$wyz$};
    
    \node (5) at (1,1) {$Bwyz$};
    \node (6) at (1,-1) {$Dyz$};
\end{scope}

\begin{scope}[every node/.style={fill=white,circle},
              every edge/.style={draw=black,very thick}]
    \path [-] (1) edge (3);
    \path [-] (2) edge (3);
    \path [-] (3) edge (4);
    \path [-] (4) edge (5);
    \path [-] (4) edge (6);
\end{scope}
\end{tikzpicture}
	\hspace{1cm}
	\begin{tikzpicture}
\begin{scope}[every node/.style={circle,thick,draw}]
    \node (C) at (-1.2,-1) {$C$};
    \node (x) at (-1.2,0) {$x$};
    \node (A) at (-1.2,1) {$A$};
    \node (y1) at (-0.4,0) {$y_1$};
    \node (y2) at (0.4,0) {$y_2$};
    \node (w) at (0,1) {$w$};
    \node (D) at (1.2,-1) {$D$};
    \node (z) at (1.2,0) {$z$};
    \node (B) at (1.2,1) {$B$};
\end{scope}

\begin{scope}[every node/.style={fill=white,circle},
              every edge/.style={draw=black,very thick}]
    \path [-] (A) edge (w);
    \path [-] (A) edge (x);
    \path [-] (A) edge (y1);
    \path [-] (B) edge (w);
    \path [-] (B) edge (y2);
    \path [-] (B) edge (z);
    \path [-] (y1) edge (y2);
    \path [-] (C) edge (x);
    \path [-] (C) edge (y1);
    \path [-] (D) edge (y2);
    \path [-] (D) edge (z);
\end{scope}
\end{tikzpicture}
	\caption{\label{fig:factor-example} When FT is run on the shown initial tensor network (left) using the shown tree decomposition (middle), \textbf{FT} produces a factored tensor network (right). Tensors of rank 3 or smaller are unchanged, and the tensor for $y$ is factored into two tensors, $y_1$ and $y_2$, each of rank 3. The factored tensor network has a contraction tree of max rank 3 while the initial tensor network only has contraction trees of max rank 4 or higher.}
\end{figure}

\begin{theorem} \label{thm:factorable-tree}
Let $N$ be a tensor network of tree-factorable tensors such that $|\tnfree{N}| \leq 3$ and the structure graph of $N$ has a tree decomposition of width $w \geq 1$.

Then for each $A \in N$ there is a tensor network $N_A$ whose contraction is $A$ that consists only of rank 3 or smaller tensors. Moreover, the disjoint union of these networks, $M = \cup_{A \in N} N_A$, is a tensor network that contracts to $\tntensor{N}$, has the same bond dimension as $N$, and has a contraction tree of max rank no larger than $\ceil{4(w+1)/3}$.
\end{theorem}
\begin{proof}
The proof proceeds in five steps: (1) compute the factored tensor network $M$, (2) construct a graph $H$ that is a simplified version of the structure graph of $M$, (3) construct a carving decomposition $S$ of $H$, (4) bound the width of $S$, and (5) use $S$ to find a contraction tree for $M$. Working with $H$ instead of directly working with the structure graph of $M$ allows us to cleanly handle tensor networks with free indices.

\textbf{Part 1: Factoring the network.}
Let $G$ be the structure graph of $N$ with all degree 0 vertices removed; $G$ must also have a tree decomposition of width $w$. Moreover, the image of $\eincf{G}: \E{G} \rightarrow 2^{\V{G}}$ is an edge clique cover of $G$. Thus using Lemma \ref{lemma:tree-simplification} we can construct a tree decomposition $(T, \chi)$ of $G$ and a bijection $g: \E{G} \rightarrow \Lv{T}$ such that $\chi \circ g = \eincf{G}$ and $width_t(T, \chi) \leq w$.

Next, for each $v \in \V{G}$, define $T_v$ to be the smallest connected component of $T$ containing $\{g(i) ~:~i \in \vinc{H}{v} \}$. Consider each $A \in N$. If $\tnfree{A} = \emptyset$, let $N_A = \{N_A\}$. Otherwise, observe that $T_A$ is a dimension tree of $A$ and so we can factor $A$ with $T_A$ using Definition \ref{def:tree-factorable} to get a tensor network $N_A$ and a bijection $g_A: \V{T_A} \rightarrow N_A$. Define $M = \cup_{A \in N} N_A$ and let $G'$ be the structure graph of $M$ with free vertex $\fv'$. The remainder of the proof is devoted to bounding the carving width of $G'$.

\textbf{Part 2: Constructing a simplified structure graph of $M$.} In order to easily characterize $G'$, we define a new, closely-related graph $H$ by taking a copy of $T_v$ for each $v \in \V{G}$ and connecting these copies where indicated by $g$. Formally, the vertices of $H$ are $\{(v, n) : v \in \V{G}, n \in \V{T_v}\}$. For every $v \in \V{G}$ and every arc in $T$ with endpoints $n, m \in \V{T_v}$, we add an edge between $(v, n)$ and $(v, m)$. Moreover, for each $e \in \E{G}$ incident to $v, w \in \V{G}$, we add an edge between $(v, g(e))$ and $(w, g(e))$. 

We will prove in Part 5 that the carving width of $G'$ is bounded from above by the carving width of $H$. We therefore focus in Part 3 and Part 4 on bounding the carving width of $H$. It is helpful for this to define the two projections $\pi_G : \V{H} \rightarrow \V{G}$ and $\pi_T : \V{H} \rightarrow \V{T}$ that indicate respectively the first or second component of a vertex of $H$. 



\textbf{Part 3. Constructing a carving decomposition $S$ of $H$.}
The idea of the construction is, for each $n \in \V{T}$, to attach the elements of $\pi_T^{-1}(n)$ as leaves along the arcs incident to $n$. To do this, for every leaf node $\ell \in \Lv{T}$ with incident arc $a \in \vinc{T}{\ell}$ define $H_{\ell, a} = \pi_T^{-1}(\ell)$. For every non-leaf node $n \in \V{T} \setminus \Lv{T}$ partition $\pi_T^{-1}(n)$ into three equally-sized sets $\{H_{n,a} : a \in \vinc{T}{n}\}$. Observe that $\{H_{n,a} : n \in \V{T}, a \in \vinc{T}{n}\}$ is a partition of $\V{H}$. 

We use this to construct a carving decomposition $S$ from $T$ by adding each element of $H_{n,a}$ as a leaf along the arc $a$. Formally, let $x_v$ denote a fresh vertex for each $v \in \V{H}$, let $y_n$ denote a fresh vertex for each $n \in \V{T}$, and let $z_{n,a}$ denote a fresh vertex for each $n \in \V{T}$ and $a \in \vinc{T}{n}$. Define $\V{S}$ to be the union of $\V{H}$ with the set of these free vertices. 

We add an arc between $v$ and $x_v$ for every $v \in \V{H}$. Moreover, for every $a \in \E{T}$ with endpoints $o, p \in \einc{T}{a}$ add an arc between $y_{o,a}$ and $y_{p,a}$. For every $n \in \V{T}$ and incident arc $a \in \vinc{T}{n}$, construct an arbitrary sequence $I_{n,a}$ from $\{x_v : v \in H_{n,a}\}$. If $H_{n,a} = \emptyset$ then add an arc between $y_n$ and $z_{n,a}$. Otherwise, add arcs between $y_n$ and the first element of $I$, between consecutive elements of $I_{n,a}$, and between the last element of $I_{n,a}$ and $z_{n,a}$. 

Finally, remove the previous leaves of $T$ from $S$. The resulting tree $S$ is a carving decomposition of $H$, since we have added all vertices of $H$ as leaves and removed the previous leaves of $T$.

\textbf{Part 4: Computing the width of $S$.} In this part, we separately bound the width of the partition induced by each of the three kinds of arcs in $S$.

First, consider an arc $b$ between some $v \in \V{H}$ and $x_v$. Since all vertices of $H$ are degree 3 or smaller, $b$ defines a partition of width at most $3 \leq \ceil{4(w+1)/3}$.

Next, consider an arc $c_a$ between $y_{o,a}$ and $y_{p,a}$ for some arc $a \in \E{T}$ with endpoints $o, p \in \einc{T}{a}$.
Observe that removing $a$ from $T$ defines a partition $\{B_o, B_p\}$ of $\V{T}$, denoted so that $o \in B_o$ and $p \in B_p$. Then removing $c_a$ from $S$ defines the partition $\{ \pi_T^{-1}(B_o), \pi_T^{-1}(B_p) \}$ of $\Lv{S}$. By construction of $H$, all edges between $\pi_T^{-1}(B_o)$ and $\pi_T^{-1}(B_p)$ are between $\pi_T^{-1}(o)$ and $\pi_T^{-1}(p)$. Since $\pi_G(\pi_T^{-1}(o)) \subseteq \chi(o)$, $\pi_G(\pi_T^{-1}(o)) \subseteq \chi(p)$, and $\pi_G$ is an injection on $\pi_T^{-1}(n)$ for all $n \in \V{T}$), it follows that the partition defined by $c_a$ has width no larger than $|\chi(o) \cap \chi(p)| \leq w+1$. 

Finally, consider an arc $d$ added as one of the sequence of $|H_{n,a}|+1$ arcs between $y_n$, $I_{n,a}$, and $z_{n,a}$ for some $n \in \V{T}$ and $a \in \vinc{T}{n}$. Some elements of $H_{n,a}$ have changed blocks from the partition defined by $c_a$. Each vertex of degree 2 that changes blocks does not affect the width of the partition, but each vertex of degree 3 that changes blocks increases the width by 1. There are at most $|H_{n,a}| \leq \ceil{(w+1)/3}$ elements of degree 3 added as leaves between $y_n$ and $z_{n,a}$. Thus the partition defined by $d$ has width at most $w + 1 + \ceil{(w+1)/3} = \ceil{4(w+1)/3}$.

It follows that the width of $S$ is at most $\ceil{4(w+1)/3}$.

\textbf{Part 5: Bounding the max-rank of $M$.} Let $\fv$ be the free vertex of the structure graph of $N$. We first construct a new graph $H'$ from $H$ by, if $\tnfree{N} \neq \emptyset$, contracting all vertices in $\pi_G^{-1}(\fv)$ to a single vertex $\fv$. If $\tnfree{N} = \emptyset$, instead add $\fv$ as a fresh degree 0 vertex to $H'$. Moreover, for all $A \in N$ with $\tdim{A} = \emptyset$ add $A$ as a degree 0 vertex to $H'$. 

Note that adding degree 0 vertices to a graph does not affect the carving width. Moreover, since $|\tnfree{N}| \leq 3$ all vertices (except at most one) of $\pi_G^{-1}(\fv)$ are degree 2 or smaller. It follows that contracting $\pi_G^{-1}(\fv)$ does not increase the carving width. Thus the carving width of $H'$ is at most $\ceil{4(w+1)/3}$.

Moreover, $H'$ and $G'$ are isomorphic. To prove this, define an isomorphism $\phi: \V{H'} \rightarrow \V{G'}$ between $H'$ and $G'$ by, for all $v \in \V{H'}$:
$$\phi(v) \equiv \begin{cases}v&\text{if}~v \in N~\text{and}~\tdim{v}=\emptyset\\\fv'&v=\fv\\g_{\pi_G(v)}(\pi_T(v))&\text{if}~v \in \V{H}~\text{and}~\pi_G(v) \in N\end{cases}$$
$\phi$ is indeed an isomorphism between $H'$ and $G'$ because the functions $g_A$ are all isomorphisms and because an edge exists between $\pi_G^{-1}(v)$ and $\pi_G^{-1}(w)$ for $v, w \in \V{G}$ if and only if there is an edge between $v$ and $w$ in $G$. Thus the carving width of $G'$ is at most $\ceil{4(w+1)/3}$. By Theorem \ref{thm:contraction-equiv-carving}, then, $M$ has a contraction tree of max rank no larger than $\ceil{4(w+1)/3}$.
\hfill$\square$
\end{proof}

The general idea of using a tree decomposition of a graph to split high-degree nodes was previously used by Markov and Shi \cite{MS11} and in the context of constraint satisfaction by Samer and Szeider \cite{SS10_2}. Both of these works focus on minimizing the treewidth of the factored graph instead of the max-rank. Translated to tensor networks (as done in Lemma 3 of \cite{oliveira18}), their constructions produce a factored network $N'$ with structure graph $G$ that satisfies all requirements of Theorem \ref{thm:factorable-tree} except with a bound of $w+1$ on the treewidth of $G$ in place of the bound on max-rank. Since the treewidth of $\Line{G}$ plus 1 is bounded by the product of the maximum degree of $G$ (namely 3) and the treewidth of $G$ plus 1 \cite{MS08}, we can then use \textbf{LG} to produce a contraction tree for $N'$ of max rank no larger than $3(w+2)$. Theorem \ref{thm:factorable-tree} thus gives an improvement on max-rank over these prior works from $3(w+2)$ to $\ceil{4(w+1)/3}$.




The construction of Theorem \ref{thm:factorable-tree} gives us the \textbf{Factor-Tree} method, which uses tree decompositions of the structure graph to preprocess the tensor network and factor high-rank tensors. See Figure \ref{fig:factor-example} for an example of the preprocessing. We show in Section \ref{sec:experiments:cachet} that \textbf{FT} can significantly improve the quality of the contraction tree on benchmarks with high-rank tensors.
\section{Implementation and Evaluation} \label{sec:experiments}
We implement Algorithm 2 in \tool{TensorOrder}, a new tool for weighted model counting using tensor networks. \tool{TensorOrder} can be configured to perform Step 2 using one of the methods from Kourtis \emph{et al.} \cite{KCMR18} (\textbf{greedy}, \textbf{metis}, and \textbf{GN}) or one of the methods presented in this paper (\textbf{LG} and \textbf{FT}). 

We use \tool{TensorOrder} to compare tensor-based methods with existing state-of-the-art tools for weighted model counting: \tool{cachet} \cite{SBK05}, \tool{miniC2D} \cite{OD15} and \tool{d4} \cite{LM17}. 
We also compare with \tool{dynQBF} \cite{CW16}, \tool{dynasp} \cite{FHMW17} and \tool{SharpSAT} \cite{Thurley2006} when the benchmarks are unweighted. Note \tool{dynQBF} and \tool{dynasp} are solvers from related domains (that can be used as model counters) that also use tree decompositions.


We compare \tool{TensorOrder} on two sets of existing benchmarks. First, in Section \ref{sec:experiments:cubic} we compare on formulas that count the number of vertex covers of randomly-generated cubic graphs \cite{KCMR18}. Second, in Section \ref{sec:experiments:cachet} we compare on formulas whose weighted count is exact inference on Bayesian networks \cite{SBK05}. 

Each experiment was run in a high-performance cluster (Linux kernel 2.6.32) using a single 2.80 GHz core of an Intel Xeon X5660 CPU and 48 GB RAM. Each implementation was run once on each benchmark with a timeout of 1000 seconds. We provide all code, benchmarks, and detailed data of benchmark runs at \url{https://github.com/vardigroup/TensorOrder}.

\subsection{Implementation Details of \tool{TensorOrder}}
\label{sec:experiments:implementation}
\tool{TensorOrder} is implemented in Python 3.6. All tensor contractions are performed using \pkg{numpy} 1.15 and 64-bit double precision floats. \tool{TensorOrder} also supports infinite-precision integer arithmetic, but the performance is significantly degraded by limited \pkg{numpy} support. Note that \pkg{numpy} is able to leverage SIMD parallelism for tensor contraction.

Both \textbf{LG} and \textbf{FT} require first finding a tree decomposition. To do this, we leverage three heuristic tree-decomposition solvers: \pkg{Tamaki} \cite{Tamaki17}, \pkg{FlowCutter} \cite{HS18}, and \pkg{htd} \cite{AMW17}. \tool{TensorOrder} therefore has three implementations of \textbf{LG} (\textbf{LG+Tamaki}, \textbf{LG+Flow}, and \textbf{LG+htd}) and three implementations of \textbf{FT} (\textbf{FT+Tamaki}, \textbf{FT+Flow}, and \textbf{FT+htd}) for different choices of solver.

All the tree-decomposition solvers we consider are online solvers and so each implementation must decide how long to run the solver (this time is included in the measured running time). \tool{TensorOrder} estimates the time to contract each potential contraction tree (using techniques from the \pkg{einsum} package of \pkg{numpy}) and continues to look for better tree decompositions until it expects to have spent more than half of the running time on finding a tree decomposition.  This strikes a balance between improving and using the contraction trees.

\begin{figure}
	\centering
	\input{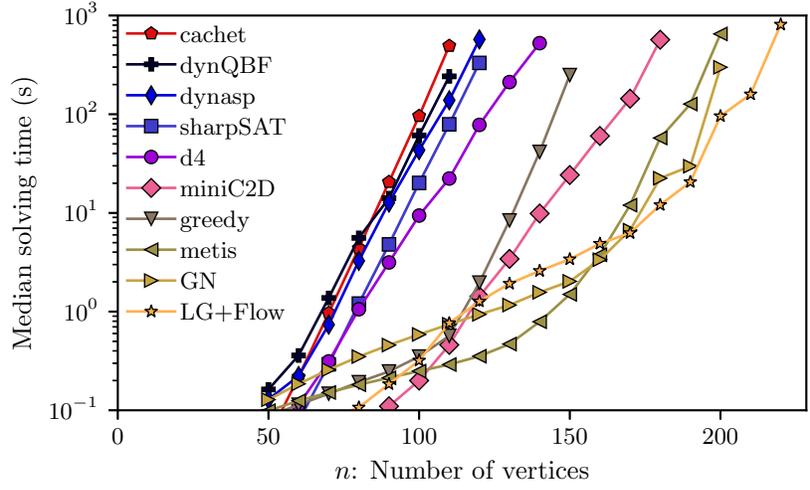}
	\input{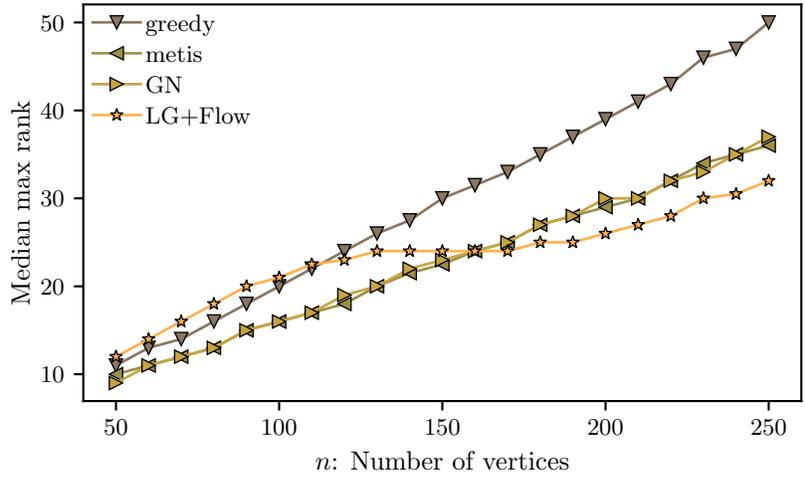}
	\caption{\label{fig:cubic-time} Median solving time (top) and max-rank of the computed contraction tree (bottom) of various model counters and tensor-based methods run on benchmarks counting the number of vertex covers of 100 cubic graphs with $n$ vertices. Solving time of datapoints that ran out of time ($1000$ seconds) or memory (48 GB) are not shown. When $n \geq 170$, our contribution \textbf{LG+Flow} is faster than all other methods and finds contraction trees of lower max-rank than all other tensor-based methods.}
\end{figure}



\subsection{Counting Vertex Covers of Cubic Graphs}
\label{sec:experiments:cubic}

We first compare on benchmarks that count the number of vertex covers of randomly-generated cubic graphs \cite{KCMR18}. In particular, for each number of vertices $n \in \{50, 60, 70, \cdots, 250\}$ we randomly sample 100 connected cubic graphs using a Monte Carlo procedure \cite{VL05}. These benchmarks are monotone 2-CNF formulas in which every variable appears 3 times.

Results on these benchmarks are summarized in Figure \ref{fig:cubic-time}. For ease of presentation, we display only the best-performing of the \textbf{LG} and \textbf{FT} implementations: \textbf{LG+Flow}. We observe that tensor-based methods are fastest when $n \geq 110$. On large graphs our contribution \textbf{LG+Flow} is fastest and able to find the lowest max-rank contraction trees. \textbf{LG+Flow} is the only implementation able to solve at least 50 benchmarks within 1000 seconds when $n$ is $220$.


\begin{figure}[t]
	\centering
	\input{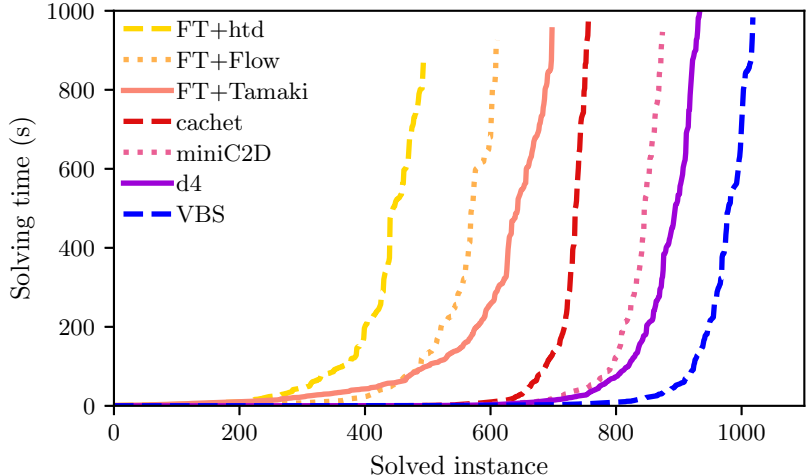}
	\caption{\label{fig:cachet-cactus} A cactus plot of the number of benchmarks solved by various methods out of 1091 probabilistic inference benchmarks. Although our contributions \textbf{FT+*} solve fewer benchmarks than the existing weighted model counters \tool{cachet}, \tool{miniC2D}, and \tool{d4}, they improve the virtual best solver on 231 benchmarks.}
\end{figure}

\subsection{Weighted Model Counting: Exact Inference}
\label{sec:experiments:cachet}
We next compare on a set of weighted model counting benchmarks from Sang, Beame, and Kautz \shortcite{SBK05}. These 1091 benchmarks are formulas whose weighted model count corresponds to exact inference on Bayesian networks. We first evaluate numerical accuracy, since our approach uses 64-bit double precision floats: on all benchmarks that \tool{miniC2D} also finishes, the weighted model count returned by our approaches agrees within $10^{-3}$.

Results on these benchmarks are summarized in Figure \ref{fig:cachet-cactus}. \textbf{FT+Tamaki} is able to solve the most benchmarks of all tensor-based methods. Our implementations of \textbf{FT} each solve fewer benchmarks than \tool{cachet}, \tool{miniC2D}, and \tool{d4}. Nevertheless, \textbf{FT+*} are together able to solve 231 benchmarks faster than existing counters (\textbf{FT+Tamaki} is fastest on 50, \textbf{FT+Flow} is fastest on 175, and \textbf{FT+htd} is fastest on 6), including 62 benchmarks on which \tool{cachet}, \tool{miniC2D}, and \tool{d4} all time out. This significantly improves the virtual best solver (VBS) when \textbf{FT+*} are included.

\begin{figure}
	\centering
	\input{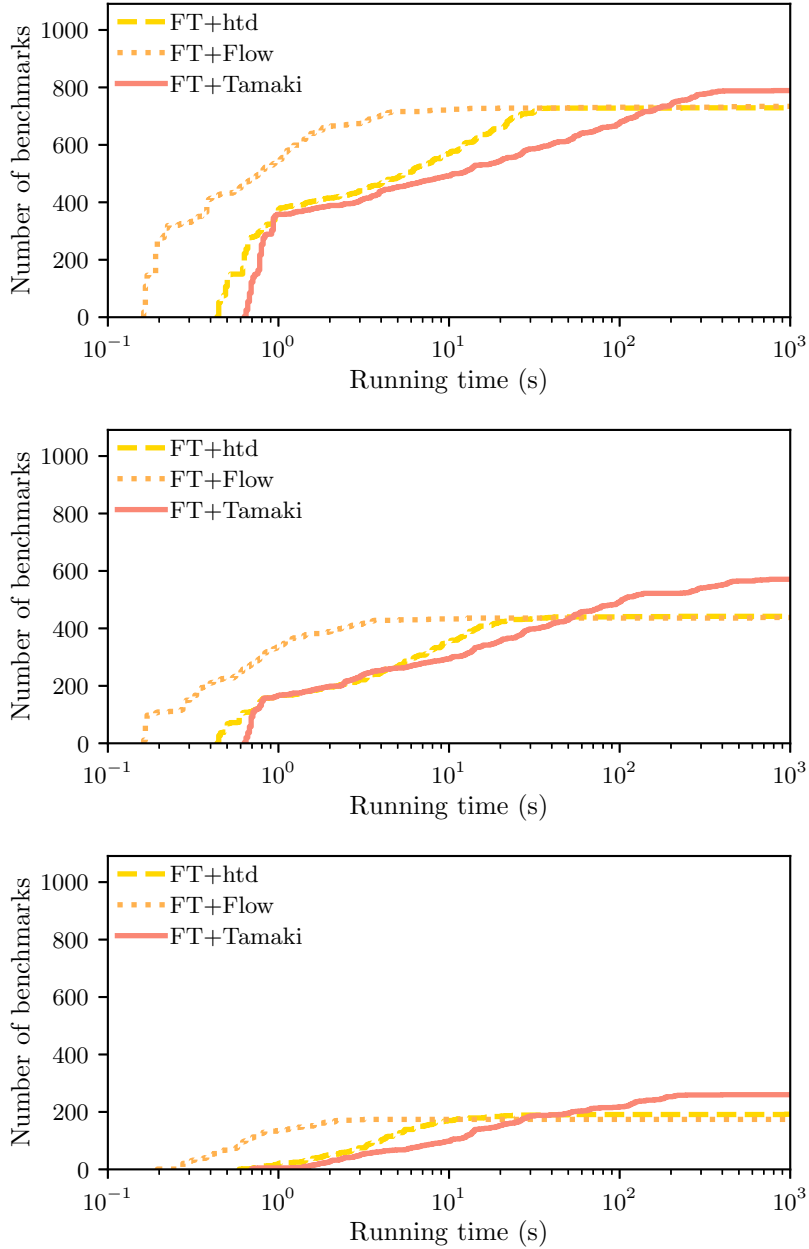}
	\caption{\label{fig:solver-analysis} The number of probabilistic-inference benchmarks (out of 1091) for which \textbf{FT+Tamaki}, \textbf{FT+Flow}, and \textbf{FT+htd} were able to find a contraction tree whose max-rank was no larger than (top) 30, (middle) 25, or (bottom) 20 within the indicated time.}
\end{figure}

The tensor-based methods (\textbf{LG}, \textbf{greedy}, \textbf{metis}, and \textbf{GN}) that do not perform factoring were only able to count a single benchmark from this set within 1000 seconds. We observe that most of these benchmarks have a variable that appears many times, which significantly hinders tensor-based methods that do not perform factoring (see Section \ref{sec:preprocessing}). 

In order to explain the relative performance of \textbf{FT+Tamaki}, \textbf{FT+Flow}, and \textbf{FT+htd} on these benchmarks, we next analyze more closely the quality of the contraction trees they produce over time. To do this, we rerun each implementation of \textbf{FT} for 1000 seconds with the contraction step (i.e. step 3 of Algorithm \ref{alg:wmc}) disabled. Each implementation of \textbf{FT} is an online solver and so produces a sequence of contraction trees over time. For each contraction tree produced on each benchmark, we record the max-rank and time of production.

Results of this experiment are summarized in Figure \ref{fig:solver-analysis}. 
\textbf{FT+Flow} is able to find more contraction trees of small max-rank within 10 seconds than the other methods, while \textbf{FT+Tamaki} is able to find more contraction trees of small max-rank within 1000 seconds than the other methods. This matches our previous observations that, among the tensor-based methods, \textbf{FT+Flow} was the fastest method on the most benchmarks while \textbf{FT+Tamaki} was able to solve the most benchmarks after 1000 seconds.


We conclude from the experiments in Section \ref{sec:experiments:cubic} and Section \ref{sec:experiments:cachet} that both \textbf{LG} and \textbf{FG} are useful as part of a portfolio of weighted model counters.

\section{Conclusions and Future Work} \label{sec:conclusion}
We presented two methods, \textbf{LG} and \textbf{FT}, for using graph decompositions to find contraction trees of small max rank of tensor networks. \textbf{LG} is a general-purpose method for finding contraction orders, while \textbf{FT} is a novel method tailored for constrained counting to handle high-rank, highly-structured tensors. We evaluated \textbf{LG} and \textbf{FT} in the context of exact weighted model counting and demonstrated that \tool{TensorOrder} is able to solve many benchmarks solved by no other exact model counter. Thus \tool{TensorOrder} is useful as part of a portfolio of weighted model counters. 

It would be interesting in the future to analyze the types of benchmarks amenable to tensor-network methods, e.g. by computing lower bounds on carving width in addition to the upper bounds given by heuristic methods. It would also be interesting to explore the impact of other preprocessing techniques (e.g., PMC \cite{LM14} or B+E \cite{LLM16}) on carving width and treewidth.

Although we restricted our experiments to a single core, a variety of libraries exist for efficiently performing tensor contractions on multiple cores or on GPUs \cite{KSTKPPRS19,NRBHHJN15}. One direction for future work is to analyze and improve the potential parallelism of tensor-based algorithms. This would allow comparison against other recent GPU-based counters, including Fichte \emph{et al.} \shortcite{FHWZ18} which also uses graph decompositions.

Tensor-based methods can also be used to count other classes of CSPs. For example, all techniques we introduced in this work would have similar performance computing the weighted model count of formulas that mix OR clauses with XOR clauses and Exactly-One clauses (as such clauses can also be represented as tree-factorable tensors). More generally, our algorithms for tensor-network contraction can be used to improve many other applications of tensor networks. Evaluating our techniques on a wider collection of tensor networks is an exciting direction for future work. 


\section*{Acknowledgment}
The authors would like to thank Mateus de Oliveira Oliveira for pointers to related work and for the alternative proof of Theorem \ref{thm:carving-equiv-tree}.

This work was supported in part by the U.S. National Science Foundation (Grants IIS-1527668, DMS-1547433, CMMI-1436845, and CMMI-1541033), by the U.S. Department of Defense (Grant W911NF-13-1-0340), and by the Big-Data Private-Cloud Research Cyberinfrastructure MRI-award funded by NSF under grant CNS-1338099,
by the Ken Kennedy Institute Computer Science \& Engineering Enhancement Fellowship funded by the Rice Oil \& Gas HPC Conference, by the Ken Kennedy Institute for Information Technology 2017/18 Cray Graduate Fellowship, 
and by Rice University.
\bibliographystyle{elsarticle-num}
\bibliography{main}

\normalsize
\newpage
\appendix

\section{A Comparison of Treewidth and Carving Width}
\label{sec:experiments:graph_analysis}

In this section, we perform an experimental comparison of treewidth and carving width of the incidence graphs of the model counting benchmarks in Section \ref{sec:experiments}.

\subsection{Experimental Setup}
On each incidence graph $G$, we ran each of the three heuristic tree decomposition solvers integrated with \tool{TensorOrder}-- \tool{Tamaki}, \tool{FlowCutter}, and \tool{htd}-- for 1000 seconds on $G$ and $\Line{G}$ and recorded the width of the best tree decomposition found amongst all tree-decomposition solvers. On each tree decomposition for $\Line{G}$ found by the solvers, we used \textbf{LG} to compute the corresponding carving decomposition of $G$ and recorded the smallest width found amongst all decompositions. Similarly, on each tree decomposition for $G$ found by the solvers, we used \textbf{FT} to compute the corresponding carving decomposition of the preprocessed graph and recorded the smallest width found amongst all decompositions. 

Unlike the experiments in Section \ref{sec:experiments}, we do not estimate the time to contract each potential contraction tree. Instead, we run each tree-decomposition solver for the full 1000 seconds on each benchmark. This allows us to more fully estimate the treewidth and carving width of these benchmarks and so more fully evaluate the potential of decomposition solvers. Each experiment was run in a high-performance cluster (Linux kernel 2.6.32) using a single 2.80 GHz core of an Intel Xeon X5660 CPU and 48 GB RAM.

\subsection{Results}
We first compare treewidth and carving width using the benchmarks from Section \ref{sec:experiments:cubic}: randomly-generated cubic graphs \cite{KCMR18}. Since \textbf{FT} only factors tensors of order 4 or higher and all vertices in each cubic graph has exactly three incident edges, \textbf{FT} performs no factoring on these graphs. Thus both \textbf{LG} and \textbf{FT} can be used to find carving decompositions.

Results on these benchmarks are summarized in Figure \ref{fig:vertex-cover-width}. We observe that, for most large graphs, the carving width of $G$ is smaller than the treewidth of $G$, which is smaller that the treewidth of $\Line{G}$. On these benchmarks, the width of the carving decompositions of $G$ found by \textbf{LG} are indeed smaller than the upper bound guaranteed by Theorem \ref{thm:carving-equiv-tree} of the width of the used tree decomposition plus one.

\begin{figure}
	\centering
	\input{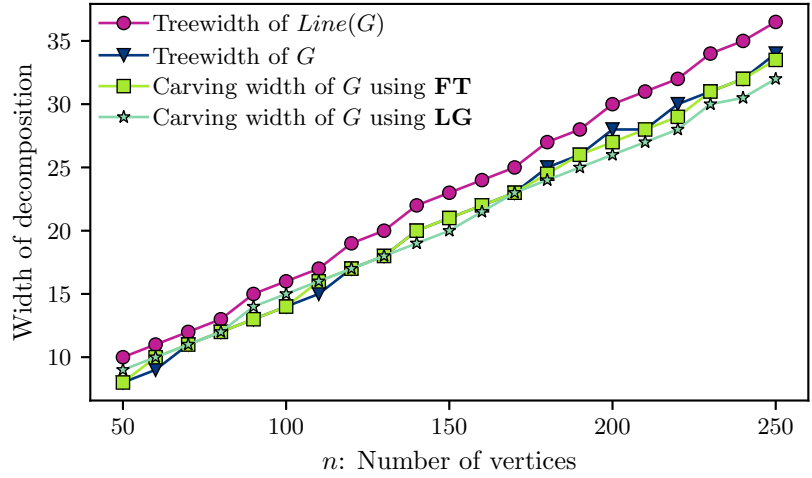}
	\caption{\label{fig:vertex-cover-width} Median of the best upper bound found for treewidth and carving width of 100 cubic graphs with $n$ vertices. Each graph-decomposition solver was run for 1000 seconds. For most large graphs, the carving width of $G$ is smaller than the treewidth of $G$, which is smaller that the treewidth of $\Line{G}$.}
\end{figure}

\begin{figure}
	\centering
	\input{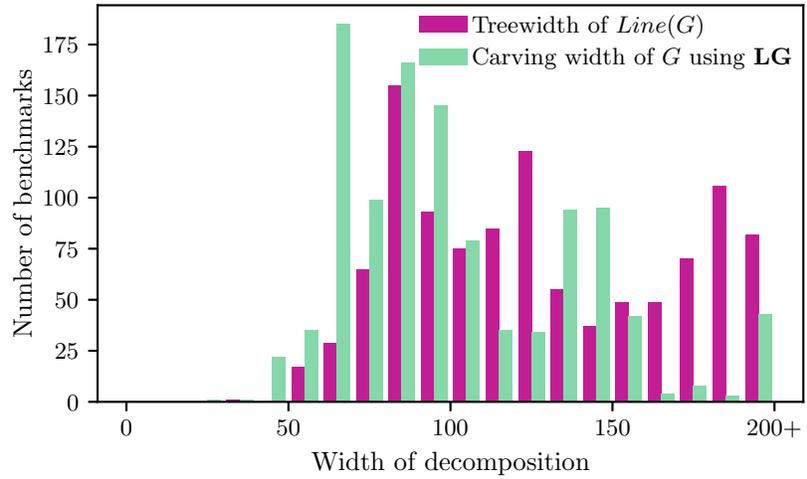}
\caption{\label{fig:wmc-width-lg} A histogram comparing an upper bound on the treewidth of $\Line{G}$ (the width of the best tree decomposition found across all solvers after 1000 seconds) with the width of the best carving decomposition of $G$ constructed by \textbf{LG}, across 1091 probabilistic inference benchmarks. On most of these benchmarks, no solver is able to find a tree decomposition of $\Line{G}$ of width smaller than 50 and so \textbf{LG} is unable to find carving decompositions of $G$ of width smaller than 50.}
\end{figure}

\begin{figure}
	\centering
	\input{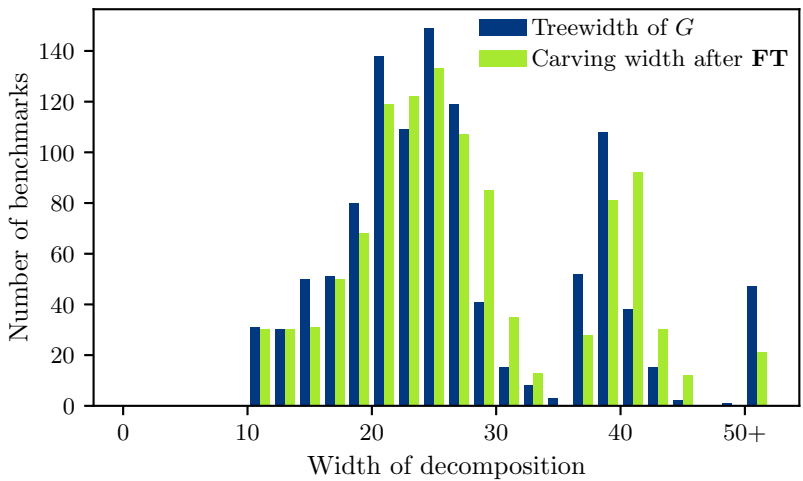}
\caption{\label{fig:wmc-width-ft} A histogram comparing an upper bound on the treewidth of $G$ (the width of the best tree decomposition found across all solvers after 1000 seconds) with the width of the best carving decomposition of preprocessed constructed by \textbf{FT}, across 1091 probabilistic inference benchmarks. On most benchmarks, the tree decompositions found of $G$ have significantly smaller width than the tree decompositions found of $\Line{G}$. Thus the carving decompositions produced by \textbf{FT} are significantly better than those produced by \textbf{LG}.}
\end{figure}

We next compare treewidth and carving width using the benchmarks from Section \ref{sec:experiments:cachet}: incidence graphs of probabilistic inference benchmarks \cite{SBK05}. Results on these benchmarks are summarized in Figure \ref{fig:wmc-width-lg} and Figure \ref{fig:wmc-width-ft}.

We observe that the carving width of $G$ found by \textbf{LG} is extremely large on these benchmarks (larger than $50$ on 1064 of the 1091 benchmarks), since these graphs have vertices of high degree. Nevertheless, the carving width of $G$ found by \textbf{LG} is still smaller than the upper bound guaranteed by Theorem \ref{thm:carving-equiv-tree} of the treewidth of $\Line{G}$ plus one on most benchmarks. We also observe that \textbf{FT} does significantly reduce the carving width by preprocessing the graph (to smaller than 50 on 1066 of the 1091 benchmarks).  Moreover, the carving width found by \textbf{FT} is smaller than the upper bound guaranteed by Theorem \ref{thm:factorable-tree} on most benchmarks.

%
\end{document}